\documentclass[conference,letterpaper]{IEEEtran}

\usepackage[utf8]{inputenc} 
\usepackage[T1]{fontenc}
\usepackage{url}              % provides \url{...}
\usepackage{cite}             % improves presentation of citations
\usepackage{amsfonts}
\usepackage{bbm}
\usepackage[cmex10]{amsmath}  % Use the [cmex10] option to ensure complicance
\usepackage{mleftright}       % fix to wrong spacing of \left-,
\mleftright                   % \middle- \right-commands 

\usepackage{graphicx}         % provides \includegraphics{...} to
\usepackage{booktabs}         % fixes poor spacing in tables and

\usepackage{amsthm}
\theoremstyle{plain}
\newtheorem{theorem}[]{Theorem}

\theoremstyle{plain}
\newtheorem{corollary}{Corollary}

\theoremstyle{definition}
\newtheorem{defn}{Definition}

\theoremstyle{remark}

\usepackage{hyperref}
\usepackage{cleveref}
\usepackage{autonum}

\DeclareMathOperator*{\argmax}{arg\,max}
\DeclareMathOperator*{\argmin}{arg\,min}

\usepackage{xcolor}
\begin{document}

\title{Single-Shot Lossy Compression for \\ Joint Inference and Reconstruction} 

%%%%%%
%\author{%
%  \IEEEauthorblockN{Anonymous Authors}
%  \IEEEauthorblockA{%
%    Please do NOT provide authors' names and affiliations\\
%    in the paper submitted for review, but keep this placeholder.\\
%    ISIT23 follows a \textbf{double-blind reviewing policy}.}
%}
%%%%% Single author, or several authors with same affiliation:
%\author{%
%  \IEEEauthorblockN{Oğuzhan Kubilay Ülger}{Elza Erkip}
%   \IEEEauthorblockA{Department of Electrical
%and Computer Engineering\\
%                    New York University\\
%                    ou2007@nyu.edu}
%                    }

% \author{%
%   \IEEEauthorblockN{Oğuzhan Kubilay Ülger}
%   \IEEEauthorblockA{Department of Electrical and Computer Engineering \\
%                     New York University\\
%                     ou2007@nyu.edu}
%   \and
%   \IEEEauthorblockN{Elza Erkip}
%   \IEEEauthorblockA{Department of Electrical and Computer Engineering\\ 
%                     New York University\\
%                     elza@nyu.edu}
% }

 \author{%
   \IEEEauthorblockN{Oğuzhan Kubilay Ülger, Elza Erkip}
   \IEEEauthorblockA{Department of Electrical and Computer Engineering \\
                    New York University, New York, USA \\
                    \{ou2007, elza\}@nyu.edu}
                   }
%
%\author{%
%   \IEEEauthorblockN{Stefan M.~Moser}
%   \IEEEauthorblockA{ETH Zürich\\
%                     ISI (D-ITET), ETH Zentrum\\
%                     8092 Zürich, Switzerland\\
%                     moser@isi.ee.ethz.ch}
%   \and
%   \IEEEauthorblockN{Albus Dumbledore and Harry Potter}
%   \IEEEauthorblockA{Hogwarts School of Witchcraft and Wizardry\\
%                     Hogwarts Castle\\ 
%                     1714 Hogsmeade, Scotland\\
%                     \{dumbledore, potter\}@hogwarts.edu}
% }

\maketitle

\begin{abstract}
  In the classical source coding problem, the compressed source is reconstructed at the decoder with respect to some distortion metric. Motivated by settings in which we are interested in more than simply reconstructing the compressed source, we investigate a single-shot compression problem where the decoder is tasked with reconstructing the original data as well as making inferences from it. Quality of inference and reconstruction is determined by a distortion criteria for each task. Given allowable distortion levels, we are interested in characterizing the probability of excess distortion. Modeling the joint inference and reconstruction problem as direct-indirect source coding one, we obtain lower and upper bounds for excess distortion probability. We specialize the converse bound and present a new easily computable achievability bound for the case where the distortion metric for reconstruction is logarithmic loss.
  
\end{abstract}

\section{Introduction}
 Classical source coding problem deals with compressing a source and transmitting it through a rate limited link. The source is then reconstructed in a possibly lossy manner where the quality of the reconstruction is measured by a distortion metric \cite{coverthomas}. The choice of distortion metric depends on properties of the source as well as the end goal of the communication. For example, if the source represents an image and the goal is to store the image for personal use, a good distortion measure would then be aligned with perceived quality by human visual system. 
 
Compression to enable inference at the decoder can be thought of as indirect lossy source coding \cite{dobrushin1962}, where the encoder cannot directly observe the information to be inferred, but only has access to a correlated observation. Continuing with the image example, consider this time that we are interested in describing the class of the image that the encoder observes. In this case, an appropriate distortion measure could be Hamming distortion. 

In this work, we are interested in compression to enable both reconstruction and inference at the decoder. This can be thought as a problem that combines the direct and indirect source coding. We can model the source as having two components: a direct component that is observed by the encoder and an indirect component which is not available at the encoder but has to be inferred from direct part. The task of the decoder is to reconstruct the observed component as well as to infer the indirect component, each with respect to its own distortion criteria. Simply compressing the observed source for reconstruction and then inferring the indirect part could lead to suboptimal performance, hence the encoder has to take both inference and reconstruction tasks into consideration when compressing the observed source. Considering the earlier example, this corresponds to compressing the image at the encoder for the decoder to both reconstruct the image and infer its underlying class. Another example, pointed in \cite{liu2022}, is speech compression where words in spoken text form have to be inferred from the compressed signal. It is natural to assume that both the speech and words are desirable at the decoder side, hence necessitating joint reconstruction and inference.

In the literature, the aforementioned problem was referred to as "semantic" \cite{liu2022} and "goal oriented" \cite{stavrou2022} source coding. In \cite{liu2022}, the authors considered the asymptotic rate-distortion region and analyzed the case with correlated Gaussian sources with squared error distortion. In \cite{stavrou2022}, the authors extended the work of \cite{liu2022} to the discrete memoryless source setting. A special case of this problem is robust descriptions \cite{gamal1982}, where the goal is to obtain two descriptions of the same observed data with two distinct distortion criteria.

Our focus in this paper is on the single-shot version of the joint inference and reconstruction problem. The single-shot compression approach is relevant in many settings. Modern image compressors that are based on artificial neural networks \cite{balle2017} work in a single-shot manner. Sensor networks with low latency and complexity requirements such as the ones deployed in autonomous cars for object detection need to operate in single-shot (or in general finite blocklength) regime. For the single-shot problem, we are interested in analyzing the probability of exceeding either distortion level under a fixed length compression rate constraint. The extension to the finite blocklength setting is also possible by considering the source as a vector. Most general bounds for the excess distortion probability for the direct and indirect source coding problems were given in \cite{kostina2012} and \cite{kostina2016noisy} respectively. The authors in \cite{shkel2017} extensively studied single-shot source coding problem considering logarithmic loss as the distortion metric. 

Here we generalize \cite{kostina2012} and \cite{kostina2016noisy} to the joint inference and reconstruction problem as well as provide specialized bounds when the distortion metric for the reconstruction of the observed source is logarithmic loss. We choose logarithmic loss because, as argued in \cite{no2015}, there is an equivalence between reconstruction of an observed source under an arbitrary distortion measure and that of under logarithmic loss. In addition, logarithmic loss is popular in learning theory and particularly useful in the settings where soft-decisions are allowed \cite{courtade2013}. In the context of source coding, it has pleasing connections with several information theoretic quantities \cite{courtade2011,courtade2013}.

% Being a generic loss function, it is of our interest to analyze the case where the reconstruction of observed data is done with respect to logarithmic loss while the other distortion function is tailored for the inference task at hand.

Rest of this paper is organized as follows. In Section \ref{sec:model} we formally define the setting of our problem and our objectives. In Section \ref{sec:bounds} we introduce achievability and converse bounds for excess distortion probability. Section \ref{sec:Logloss} specializes the bounds for the logarithmic loss setting as well as providing a new achievability bound. Finally we finish the paper with concluding remarks in Section \ref{sec:conc}.

\textit{Notation:} Random variables and their realizations are denoted by capital and small letters respectively (e.g. $X$ and $x$). For a random variable $X$, the alphabets are represented by caligraphic letters such as $\mathcal{X}$. Distributions and conditional distributions are denoted by $P_X$ and $P_{X|Y}$ respectively. Cardinality of a set $\mathcal{X}$ is denoted by $|\mathcal{X}|$. The set of probability distributions on $\mathcal{X}$ is denoted by $\mathcal{P}_\mathcal{X}$. All logarithms are base 2.

\section{System Model} \label{sec:model}

 Let $(X,Y)$ be a pair of discrete random variables with joint distribution $P_{XY}$ defined on the finite product alphabet $\mathcal{X} \times \mathcal{Y}$. Here $X$ represents the direct component of the source observed  at the encoder and $Y$ is the indirect component. The goal is to reconstruct $X$ and infer $Y$ from a compressed version of $X$. The block diagram for our system model can be seen in Figure~\ref{fig:model}.

The encoder is tasked with mapping the observed source into $M$ codewords and the decoder wishes to reconstruct the pair $(\hat{X},\hat{Y})$ defined on reconstruction alphabets $\hat{\mathcal{X}}\times\hat{\mathcal{Y}}$. The quality of the reconstructions are determined by two distortion criteria $d_1:\mathcal{X}\times\hat{\mathcal{X}} \longrightarrow \mathbb{R}^+$ and $d_2:\mathcal{Y}\times\hat{\mathcal{Y}} \longrightarrow \mathbb{R}^+$. These distortion measures are referred to as direct (for reconstruction) and indirect (for inference) distortion measure respectively. Given two allowed distortion levels $(D_1,D_2)$ we are interested in keeping the probability of exceeding either distortion level below some threshold $\epsilon$
\begin{align}
    \mathbb{P}[\{d_1(X,\hat{X})>D_1\} \cup \{d_2(Y,\hat{Y})>D_2\}]\leq\epsilon.
\end{align}
Following definitions will make our goals formal.

\begin{defn}
A \textit{single-shot joint inference and reconstruction code} $(f,g)$ with size $M$ consists of pair of mappings:
\begin{align}
    f &: \mathcal{X} \longrightarrow \{1,\ldots,M\} \\
    g &: \{1,\ldots,M\} \longrightarrow \hat{\mathcal{X}} \times \hat{\mathcal{Y}}.
\end{align}
where $f$ represents the encoder and $g$ the joint decoder consisting of two elements $g=(g_1,g_2)$.
\end{defn}
\begin{defn} \label{defn:code}
Given a direct source $X$ and an indirect source $Y$ with joint distribution $P_{XY}$, reconstruction alphabets $\hat{\mathcal{X}}\times\hat{\mathcal{Y}}$ and two distortion metrics $d_1:\mathcal{X}\times\hat{\mathcal{X}} \longrightarrow \mathbb{R}^+$ and $d_2:\mathcal{Y}\times\hat{\mathcal{Y}} \longrightarrow \mathbb{R}^+$, \textit{an $(M,D_1,D_2,\epsilon)$-code} is a single-shot joint inference and reconstruction code with mappings $(f,g)$ with size $M$ such that:
\begin{align}
    \mathbb{P}[\{d_1(X,g_1(f(X)))>D_1\} \cup \{d_2(Y,g_2(f(X)))>D_2\}] \leq \epsilon.
\end{align}
\noindent where $D_1,D_2$ are given distortion levels and $\epsilon$ is the excess distortion probability threshold.
\end{defn}
\begin{defn}
    Given $(M,D_1,D_2)$, \textit{the minimum achievable excess distortion probability} is defined as:
\begin{align}
    \epsilon^*(M,D_1,D_2) = \min \{\epsilon : \exists (M,D_1,D_2,\epsilon)\text{-code}\}.
\end{align}
\end{defn}
In this work, we are interested in characterizing $\epsilon^*(M,D_1,D_2)$ by providing achievability and converse bounds.

\begin{figure}[t]
    \centering
    \includegraphics[width=0.44\textwidth]{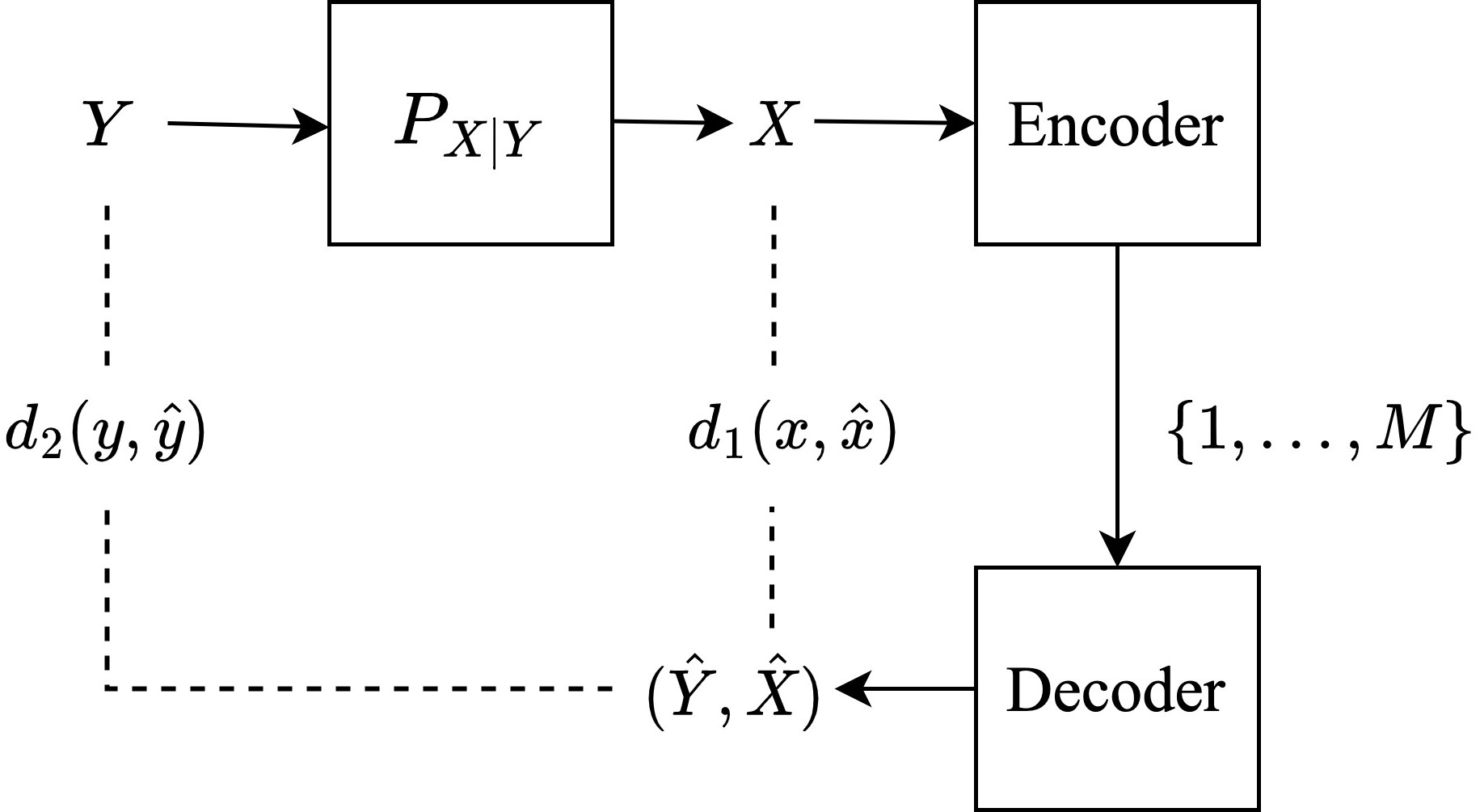}
    % where a .pdf suffix will be assumed for pdflatex
    \caption{Block diagram for single-shot joint inference and reconstruction problem}
    \label{fig:model}
    \vspace{-0.3cm}
\end{figure}
\vspace{-0.3cm}
\subsection{Preliminaries} 
We first present asymptotic rate distortion functions (when $n$ i.i.d. source samples are jointly compressed with $n\rightarrow \infty$) for the three related source coding settings. Given a source distribution $P_X$ and direct distortion metric $d_1(x,\hat{x})$ the well-known direct rate distortion function \cite{coverthomas} is given by
\begin{align}
\begin{split}
    R_1(D_1) &= \min_{P_{\hat{X}|X}:}I(X;\hat{X}) \label{eq:r1d1}\\
    &\mathbb{E}[d_1(X,\hat{X})] \leq D_1
\end{split}
\end{align}
For given $P_{XY}$ where $X$ is the observed source and Y is the indirect source with distortion metric $d_2(y,\hat{y})$, the indirect rate distortion function \cite{berger2003} is given by
\begin{align}
\begin{split}
    R_2(D_2) &= \min_{P_{\hat{Y}|X}:}I(X;\hat{Y}) \label{eq:r2d2} \\ 
     &\mathbb{E}[\Tilde{d}_2(X,\hat{Y})] \leq D_2
 \end{split}
\end{align}
where $\Tilde{d}_2(x,\hat{y}) = \mathbb{E}[d_2(Y,\hat{y})|X=x]$.
For the asymptotic case of the joint inference and reconstruction problem, the  rate distortion function is given by \cite{liu2022,stavrou2022}
\begin{align}
\begin{split}
    R(D_1,D_2) &= \min_{P_{\hat{X}\hat{Y}|X}: } I(X;\hat{X},\hat{Y}) \label{eq:RD} \\
    &\mathbb{E}[d_1(X,\hat{X})] \leq D_1 \\
    &\mathbb{E}[\Tilde{d}_2(X,\hat{Y})] \leq D_2.
\end{split}
\end{align}

In the joint problem we denote the minimum achievable distortions with finite rate as:
\begin{align}
    D_{1,\min} &= \min\{D_1 : R(D_1,D_2 = \infty) < \infty\} \\
    D_{2,\min} &= \min\{D_2 : R(D_1 = \infty ,D_2 ) < \infty \}
\end{align}

Although in the single-shot setting source coding with expected distortion constraint has a different solution than the one with probability of excess distortion constraint, for bounded distortion measures, asymptotically minimum achievable rates under vanishing excess distortion constraint and expected distortion constraint are the same \cite[Theorem 7.3]{csiszar2011}.

 Finally we define some information theoretic quantities that will be useful to us in providing converse bounds. 
%  \begin{defn} Given a discrete random variable $X$ with distribution $P_X$ defined on $\mathcal{X}$. The \textit{information} of $x$ is:
%  \begin{align}
%     \imath_X(x) = \log \frac{1}{P_X(x)}.
% \end{align}    
%  \end{defn}
\begin{defn}
Given a pair of discrete random variables $(X,Y)$ with distribution $P_{XY}$, the \textit{mutual information density} is:
\begin{align}
    \imath_{X;Y}(x;y) = \log \frac{P_{X|Y}(x|y)}{P_X(x)}.
\end{align}
\end{defn}
\noindent and $\imath_X(x):=\imath_{X;X}(x,x)$ is the information of $x$.
\begin{defn} \cite[Definition 1]{kostina2016noisy} 
     \textit{The indirect $D_2$-tilted information} is defined as:
\begin{align}
    \jmath_{X;\hat{Y}^\star}(x,y,\hat{y},D_2) &= \imath_{X;\hat{Y}^{\star}}(x;\hat{y}) +\lambda_2'(d_2(y,\hat{y})-D_2)
\end{align}
where $\imath_{X;\hat{Y}^{\star}}(x;\hat{y})$ is defined with respect to $P_{\hat{Y}^{\star}|X}$, a distribution achieving the minimum in (\ref{eq:r2d2}), and
\begin{align}
    \lambda_2' = - \left.\frac{d R_2(D_2')}{d D_2'}\right|_{D_2'=D_2}.
\end{align}
assuming $R_2(D_2)$ is differentiable at $D_2$.
\end{defn}
We adopt the above definition to our joint inference and reconstruction problem as follows:
\begin{defn} \label{defntilt}
     \textit{The joint $(D_1,D_2)$-tilted information} is defined as:
\begin{align}
\begin{split}
    \jmath_{X;\hat{X}^*\hat{Y}^*}(x,y,\hat{x},\hat{y},D_1,D_2) &= \imath_{X;\hat{X}^*\hat{Y}^*}(x;\hat{x},\hat{y}) \\
    &+ \lambda_1(d_1(x,\hat{x})-D_1) \\
    &+\lambda_2(d_2(y,\hat{y})-D_2)
\end{split}
\end{align}
where $\imath_{X;\hat{X}^*\hat{Y}^*}(x;\hat{x},\hat{y})$ is defined with respect to  $P_{\hat{X}^*\hat{Y}^*|X}$, a distribution that achieves the minimum in the rate distortion
function defined in (\ref{eq:RD}), and
\begin{align}
\begin{split} \label{eq:lamdas}
    \lambda_1 &= - \left.\frac{\partial R(D_1',D_2)}{\partial D_1'}\right|_{D_1'=D_1}, \\
    \lambda_2 &= - \left.\frac{\partial R(D_1,D_2')}{\partial D_2'}\right|_{D_2'=D_2}.
\end{split} 
\end{align}
assuming $R(D_1,D_2)$ is differentiable at $(D_1,D_2)$.
\end{defn}

\section{Single-Shot Bounds} \label{sec:bounds}
In this section we introduce single-shot achievability and converse bounds for the minimum excess distortion probability $\epsilon^*(M,D_1,D_2)$. In Section \ref{sec:Logloss}, we specialize the converse bound to the case where $d_1(x,\hat{x})$ is logarithmic loss and provide an alternate achievability bound for this case.

\subsection{Achievability Bound}

\begin{theorem} \label{ach}
For the single-shot joint inference and reconstruction problem with direct source $X$ and indirect source $Y$ with joint distribution $P_{XY}$ and distortion metrics $d_1(x,\hat{x})$ and $d_2(y,\hat{y})$ respectively, the minimum achievable excess distortion probability satisfies:
\begin{align}
    \epsilon^*(M,D_1,D_2) \leq \inf_{P_{\hat{X}\hat{Y}}} \int_{0}^{1} \mathbb{E} \left[ \left. \mathbb{P} \left[\pi(X,\hat{X},\hat{Y})>t \right| X \right] ^M \right]dt \label{eq:ach}
\end{align}
where the infimum is taken over all distributions $P_{\hat{X}\hat{Y}}$ independent of $X$ and
\begin{align}
    \pi(x,\hat{x},\hat{y}) = \mathbb{P}\left[ \{d_1(X,\hat{x}) >D_1\} \cup \{d_2(Y,\hat{y}) > D_2\} | X=x \right]. \label{EqPi}
\end{align}
\end{theorem}
\begin{proof}
    See Appendix \ref{apx:ach}.
\end{proof}
The proof is based on random coding arguments as in \cite{kostina2012,kostina2016noisy} and extends their result to the joint direct and indirect distortion constraints case. In the special case where $D_1=\infty$, meaning that there is only a single distortion constraint on the indirect part of the source, Theorem \ref{ach} reduces to achievability result in \cite{kostina2016noisy}. This can be seen by setting $D_1=\infty$ in (\ref{EqPi}). Note that in this case a result we get a union of an empty set and $\{d_2(Y,\hat{y}) > D_2\}$. Hence we remove the dependency on $d_1(x,\hat{x})$ and $\hat{X}$.

%In the special case where $Y=X$, the joint problem reduces to the robust descriptions problem \cite{gamal1982} where the decoder outputs two versions of the source, $\hat{X}_1$ and $\hat{X}_2$ each with its own distortion metric $d_1(x,\hat{x}_1)$ and $d_2(x,\hat{x}_2)$. For the single-shot robust descriptions problem, Theorem \ref{ach} simplifies to the following corollary:

%\begin{corollary} \label{cor1}
 %For the single-shot robust descriptions problem with source distribution $P_X$ and two distortion metrics $d_1(x,\hat{x}_1)$ and d_2(x,\hat{x}_2)$, the minimum excess distortion probability is bounded above by:
%\begin{align}
%\begin{split}
%    &\epsilon^*(M,D_1,D_2) \leq \\ 
%    &\inf_{P_{\hat{X}_1\hat{X}_2}} \mathbb{E}\left[  \mathbb{P}\left[\{d_1(X,\hat{X}_1) > D_1\} \cup \{d_2(X,\hat{X}_2)>D_2 \}|X \right]^M \right]
%\end{split}
%\end{align}
%\end{corollary}
%\begin{proof}
%    See Appendix \ref{apx:cor1}
%\end{proof}
% \begin{proof}
% Setting $Y=X$ in (\ref{EqPi}) gives us:
% \begin{align}
%     \pi(x,\hat{x}_1,\hat{x}_2)  = \mathbbm{1}\left[\{d_1(x,\hat{x})> D_1 \} \cup \{ d_2(x,\hat{x}_2) > D_2\}\right]
% \end{align}
% Then for $0\leq t < 1$
% \begin{align}
%     \mathbb{P}\left[ \left. \pi(X,\hat{X}_1,\hat{X}_2)>t\right|X \right] =
%     \mathbb{E}\left[ \left. \pi(X,\hat{X}_1,\hat{X}_2) \right| X\right].
% \end{align}
% This allows us to simplify the integral in (\ref{eq:ach}) and obtain the desired expression.
% \end{proof}

\subsection{Converse Bound}

\begin{theorem} \label{conv}
For the single-shot joint inference and reconstruction problem with direct source $X$ and indirect source $Y$ with joint distribution $P_{XY}$ and distortion metrics $d_1(x,\hat{x})$ and $d_2(y,\hat{y})$ respectively, the minimum achievable excess distortion probability satisfies:
\begin{align}
\begin{split}
    &\epsilon^*(M,D_1,D_2) \geq \inf_{P_{\hat{X}\hat{Y}|X}} \sup_{\gamma \geq 0 } \\ 
    &\left\{ \mathbb{P}\left[\jmath_{X;\hat{X}^*,\hat{Y}^*}(X,Y,\hat{X},\hat{Y},D_1,D_2) \geq \log M + \gamma \right] - 2^{-\gamma} \right\}    
\end{split}
\end{align}
where the joint tilted information $\jmath_{X;\hat{X}^*,\hat{Y}^*}(x,y,\hat{x},\hat{y},D_1,D_2)$ is defined according to some distribution $P_{\hat{X}^*\hat{Y}^*|X}$ that achieves the asymptotic rate-distortion function $R(D_1,D_2)$.
\end{theorem}

\begin{proof}
See Appendix \ref{apx:conv}.
\end{proof}
Note that if there are multiple distributions achieving the joint rate distortion function $R(D_1,D_2)$, then there are multiple joint $(D_1,D_2)$-tilted information functions and the converse bound in Theorem \ref{conv} holds for any of those. Furthermore, we can obtain a similar bound by replacing $\lambda_1$ and $\lambda_2$ in Definition \ref{defntilt} with arbitrary values. Hence, one can also obtain a tighter converse that includes these as optimization parameters similar to that of \cite[Theorem 2]{kostina2016noisy}. However, the converse bound presented here is simpler as it has less parameters to optimize and can be further simplified for the case of logarithmic loss as shown in Section \ref{sec:Logloss}.

\section{Logarithmic Loss for Reconstruction} \label{sec:Logloss}
Theorems~\ref{ach} and~\ref{conv} provide general achievability and converse results for arbitrary distortion measures. The converse result makes use of the asymptotic rate distortion function defined in (\ref{eq:RD}). Notice that from the expressions  (\ref{eq:r1d1}), (\ref{eq:r2d2}) and (\ref{eq:RD}), it is immediately clear that $R(D_1,D_2) \geq \max \left( R_1(D_1),R_2(D_2)\right)$. Below we show that when the direct distortion measure is logarithmic loss, i.e.
\begin{align} \label{logeq}
    d_1(x,\hat{x}) = \log \frac{1}{\hat{x}(x)}
\end{align}
where $\hat{\mathcal{X}} = \mathcal{P}_\mathcal{X}$ (reconstructions are probability distributions), then $R(D_1,D_2)=\max(R_1(D_1),R_2(D_2))$ in a large region of $(D_1,D_2)$. The rate distortion function for logarithmic loss is given by \cite{courtade2013}:
\begin{align}
    R_1(D_1) = H(X) - D_1
\end{align}
So any $P_{\hat{X}|X}$ distribution with $H(\hat{X}|X)=D_1$ achieves the rate distortion limit. This linear relation and flexibility are the properties of the logarithmic loss exploited in proving the following theorem.
\begin{theorem} \label{thmLogLoss}
Consider the joint inference and reconstruction problem in the asymptotic regime with $d_1(x,\hat{x})$ as logarithmic loss defined in (\ref{logeq}) and $d_2(y,\hat{y})$ with an arbitrary distortion measure. Then for all $D_2\geq D_{2,\min}$ and $D_1$ satisfying:
\begin{align}
    R_2(D_{2,\min}) \geq R_1(D_1) = H(X) - D_1,
\end{align}
we have
\begin{align}
    R(D_1,D_2) = \max(R_1(D_1),R_2(D_2)).
\end{align}
and the rate distortion achieving $\hat{X}$ and $\hat{Y}$  are functions of each other.
\end{theorem}

\begin{proof}
    See Appendix \ref{apx:loglossUniv}.
\end{proof}

Note that if $R_2(D_{2,\min}) = H(X)$, then the theorem holds for all $(D_1,D_2)$. A similar property for logarithmic loss was shown in \cite{no2015}, where the authors prove any discrete memoryless source is successively refinable when the distortion measure is logarithmic loss.

\begin{corollary} \label{convLogLoss}
    Consider the single-shot joint inference and reconstruction problem with direct source $X$ and indirect source $Y$ with joint distribution $P_{XY}$. Suppose that direct distortion measure $d_1(x,\hat{x})$ is logarithmic loss and the indirect
    distortion measure $d_2(y,\hat{y})$ is arbitrary. Further assume that $(D_1,D_2)$ satisfies the condition for Theorem~\ref{thmLogLoss}. Then for $R_2(D_2) < H(X) - D_1$,
\begin{align}
\epsilon^*(M,D_1,D_2) \geq \sup_{\gamma\geq0}\left\{ \mathbb{P}\left[ \imath_X(X) \geq D_1 + \log M + \gamma \right] - 2^{-\gamma}\right\} \label{eq:cor2p1}
\vspace{-0.7cm}
\end{align}
and for $R_2(D_2) \geq H(X) - D_1$,
\begin{align}
\begin{split}
    \epsilon^*(M,D_1,&D_2) \geq \inf_{P_{\hat{Y}|X}}\sup_{\gamma\geq0} \\
    &\left\{ \mathbb{P}\left[ \jmath_{X;\hat{Y}^*}(X,Y,\hat{Y},D_2) \geq \log M + \gamma \right] - 2^{-\gamma}\right\}. \label{eq:cor2p2}
\end{split}
\end{align}
\end{corollary}
\begin{proof}
See Appendix \ref{apx:cor2}.
\end{proof}
The converse bound for logarithmic loss presented in Corollary \ref{convLogLoss} follows from the properties of the asymptotic rate distortion function in Theorem \ref{thmLogLoss}. However, for single-shot achievability it is not directly obvious how Theorem \ref{thmLogLoss} could be useful. Furthermore, specializing Theorem \ref{ach} to logarithmic loss is not straightforward. Below we first introduce a new achievability bound for the indirect single-shot lossy source coding problem (Theorem \ref{Ach2}) and then extend it to joint inference and reconstruction where the direct distortion metric is logarithmic loss (Theorem \ref{thm:achLog}).

\begin{theorem} \label{Ach2}
    Consider an observed source $X$ and an indirect source $Y$ with joint distribution $P_{XY}$. Then for any $P_{\hat{Y}}$ and $0\leq\epsilon'\leq1$, there exists an $(M,D_1=\infty,D_2,\epsilon)$-code such that
    \begin{align}
        \epsilon = \mathbb{P}[d_2(Y,\hat{Y})>D_2] \leq \epsilon'\left(1-\mathbb{E}\left[\eta(\epsilon')^M\right]\right)+\mathbb{E}\left[\eta(\epsilon')^M \right]
    \end{align}
    where the expectation is with respect to marginal distribution of $X$ and
    \begin{align}
        \eta(\epsilon') &= \mathbb{P}[\pi'(X,\hat{Y})>\epsilon'|X] \label{eq:eta}\\
        \pi'(x,\hat{y}) &= \mathbb{P}[d_2(Y,\hat{y})>D_2|X=x].
    \end{align}
\end{theorem}

\begin{proof}
    See Appendix \ref{apx:alt}.
\end{proof}
The difference between Theorem \ref{ach} (specialized to $D_1=\infty$) and Theorem \ref{Ach2} lies in the design of encoders in their proofs. The encoder in Theorem \ref{ach} maps each symbol $x$ to the codeword that minimizes probability of excess distortion, while the encoder of Theorem \ref{Ach2} maps the each symbol to any of the codewords for which probability of excess distortion is below some threshold $\epsilon'$. In the special case of $Y=X$, the encoders are equivalent to each other.
\begin{theorem} \label{thm:achLog}
Consider the single-shot joint inference and reconstruction problem with direct source $X$ and indirect source $Y$ with joint distribution $P_{XY}$. Suppose that direct distortion measure $d_1(x,\hat{x})$ is logarithmic loss and the indirect distortion measure $d_2(y,\hat{y})$ is arbitrary. Then
    \begin{align}
    \begin{split}
        \epsilon^*(M,D_1,D_2) &\leq \inf_{P_{\hat{Y}}}\inf_{\gamma\geq0}\inf_{0\leq\epsilon'\leq1} \left\{ \epsilon'\left(1-\mathbb{E}\left[\eta(\epsilon')^M\right]\right) \right. \\
        &+ \mathbb{E}\left[\eta(\epsilon')^M \right](1+2^{1-\gamma})  \\
        &+ 2^{1-\gamma}\sum_{k=1}^{M} {M \choose k}\frac{M}{k} \mathbb{E}[\eta(\epsilon')^{M-k}(1-\eta(\epsilon'))^k] \\
        &+ \mathbb{P}[\imath_X(X) > D_1 + \log M - \gamma] \} \label{eq:thmlast}
    \end{split}
    \end{align}
where $\eta(\epsilon')$ is defined in (\ref{eq:eta}).
\end{theorem}
\begin{proof}
    See Appendix \ref{apx:achLog}.
\end{proof}
Optimization over the joint distribution $P_{\hat{X}\hat{Y}}$ in Theorem \ref{ach} can be challenging due to the fact that $\hat{\mathcal{X}}=\mathcal{P}_{\mathcal{X}}$, the set of all probability distributions on $\mathcal{X}$. Specifically, considering the case where $|\mathcal{X}|$ is large compared to $|\mathcal{Y}|$, which holds in most of the joint inference and reconstruction problems such as the image example, the reduction of optimization over just $P_{\hat{Y}}$ would be beneficial. We also note that in when $D_1$ and $M$ are small, the last term in (\ref{eq:thmlast}) dominates and it is close to 1, hence the bound in Theorem \ref{thm:achLog} is more useful when the distortion constraint on the source $X$ is not that strict.

\subsection{Example}
We illustrate lower and upper bounds in Corollary \ref{convLogLoss} and Theorem \ref{thm:achLog} through an example. Our example further highlights the computational advantage of Theorem \ref{thm:achLog}. Suppose that $Y\sim\mathrm{Uniform}\{0,\ldots,m-1\}$ and for a $0\leq p\leq1$, $P_{X|Y}$ is given as
\begin{align}
    P_{X|Y}(x|y) = \begin{cases}
        \phi(x-yn), &x=y(n+1),\ldots,y(n+1)+n  \\
        0, &\mathrm{otherwise} \label{eq:expl}
    \end{cases}
\end{align}
where $n>1$ is an integer and
\begin{align}
    \phi(x) = \begin{cases}
        {n \choose x}p^x (1-p)^{n-x}, &x=0,\ldots,n \\
        0, &\mathrm{otherwise}.
    \end{cases}
\end{align}
We assume that $d_1(x,\hat{x})$ is logarithmic loss as in (\ref{logeq}) and $d_2(y,\hat{y})= \mathbbm{1}\{y\neq \hat{y}\}$ is the Hamming distortion. 
Note that $X$ is a Binomial($n,p$) random variable whose alphabet is determined by its class, represented by $Y$. Also note that $Y$ is a deterministic function of $X$. Due to this deterministic relationship the indirect rate distortion function $R_2(D_2)$ simplifies to a rate distortion function of a uniform discrete random variable under Hamming distortion\cite{coverthomas}:
\begin{align}
    R_2(D_2) = \log m - h(D_2) - D_2 \log(m-1)
\end{align}
where $h(D_2)=-D_2\log D_2 - (1-D_2)\log (1-D_2)$ is the binary entropy function. The converse bound in (\ref{eq:cor2p2}) simplifies to \cite{kostina2012}:
\begin{align}
    \epsilon^*(M,D_1,D_2) \geq 1 - M/m.
\end{align}

% Note that for (\ref{eq:cor2p2}) to hold $D_1$ needs to be large due to the assumption $R_2(D_2)\geq H(X) - D_1$. 
We pick $\hat{Y}\sim \mathrm{Uniform}\{0,\ldots,m-1\}$ in Theorem \ref{thm:achLog} due to the symmetry of the problem. It is straightforward to show that for any $\epsilon'<1$ and $D_2<1$ we have $\eta(\epsilon')=(m-1)/m$. Hence the infimum in (\ref{eq:thmlast}) is achieved by $\epsilon'=0$ and the achievability simplifies to
\begin{align}
\begin{split}
    \epsilon^*(M,D_1,D_2) &\leq \inf_{\gamma\geq0} \left\{ \left(\frac{m-1}{m}\right)^M(1+2^{1-\gamma}) \right.  \\ 
    +&\left. 2^{1-\gamma}\sum_{k=1}^{M} {M \choose k}\frac{M}{k} \left(\frac{m-1}{m}\right)^{M-k}\left(\frac{m-1}{m}\right)^k \right. \\ 
    +& \mathbb{P}[\imath_X(X) > D_1 + \log M - \gamma] \Biggl\}
\end{split}
\end{align}
 Figure ~\ref{fig:numResult}  shows aforementioned bounds against $M$ for $m=10$, $n=6$, $p=0.1$ $D_1=6$ and $D_2<1$ so that $R_2(D_2) \geq H(X) - D_1$. It can be seen that we are able to obtain nontrivial bounds for a large portion of $M$ values. Notice that although $|\mathcal{X}|=70$ and $|\mathcal{Y}|=10$ the numerical calculation of the bounds is straightforward unlike the general bound in Theorem \ref{ach} which contains infimum over $P_{\hat{X}\hat{Y}}$.
\begin{figure}[t]
    \centering
    \includegraphics[width=0.48\textwidth]{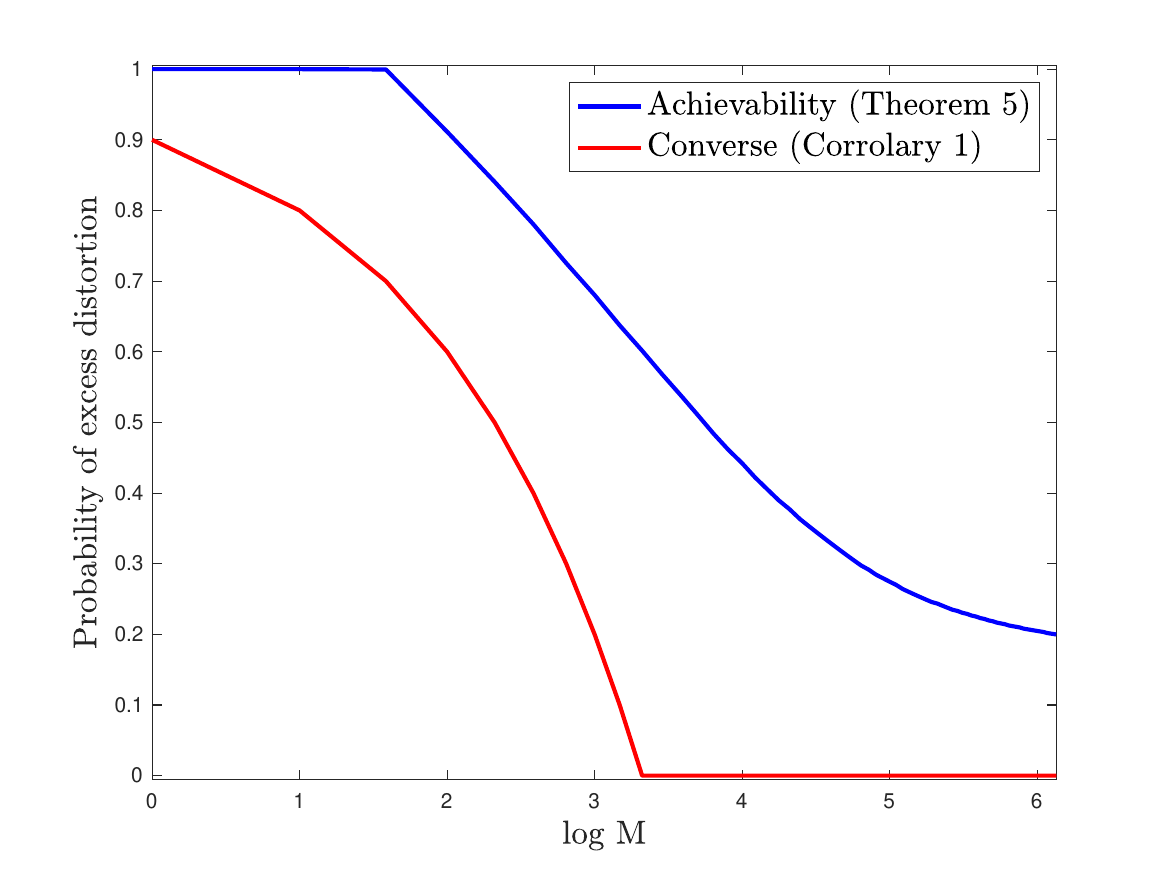}
    % where a .pdf suffix will be assumed for pdflatex
    \caption{Bounds for minimum excess distortion probability $\epsilon^*(M,D_1,D_2)$ where $Y\sim\mathrm{Uniform}\{0,\ldots,9\}$ and distribution of $X$ is given in (\ref{eq:expl}) with $n=6$, $p=0.1$, $D_1=6$, $D_2<1$.}
    \label{fig:numResult}
\end{figure}

\section{Conclusion} \label{sec:conc}
In this work, we have studied the joint inference and reconstruction problem, and characterized upper and lower bounds to the excess distortion probability. We have obtained a simplified converse bound for the special case when the reconstruction distortion measure is logarithmic loss. Using an alternate achievability for the indirect single-shot source coding problem, we have obtained a specialized achievability bound for joint inference and reconstruction problem where the reconstruction is under logarithmic loss. Finally, we have provided a numerical example highlighting computational advantage of this new achievability bound. One possible avenue for future exploration is to obtain a tighter bound without compromising on the computational benefit.

\appendices
\section{Proof of Theorem \ref{ach}} \label{apx:ach}
%%%%%%%%%%%%%%%%%%%%%%%%%%%%%%%%%%%%%%%%%%%%%%%%
    Fix a $P_{\hat{X}\hat{Y}}$ on $\hat{\mathcal{X}}\times \hat{\mathcal{Y}}$. Pick a random codebook with $M$ codeword pairs $\{(c_{1,1},c_{1,2}),\ldots,(c_{M,1},c_{M,2})\}$ generated independently from $P_{\hat{X}\hat{Y}}$. Given a source realization $x$, encoder simply sends the index of the codeword $f(x) = i^*$
    \begin{align}
        i^* \in \argmin_{i \in {1,\ldots,M}} \pi(x,c_{i,1},c_{i,2}) \label{enc1}
    \end{align}
    where ties are broken arbitrarily. After receiving the index, joint decoder outputs the codeword pair with given index $g(i)= (g_1(i),g_2(i))=(c_{i,1},c_{i,2})$. For the sake of simplicity we will use $\pi(x,g(f(x)))$ in place of $\pi(x,g_1(f(x)),g_2(f(x)))$. Denote the excess distortion event:
    \begin{align}
        E : \{d_1(x,g_1(f(x)))>D_1\} \cup \{d_2(y,g_2(f(x)))>D_2\}.
    \end{align}
     We can write the excess distortion probability of this scheme as follows:
    \begin{align}
        \mathbb{P}[E] &= \mathbb{E}[\pi(X,g(f(X)))] \\
        &= \mathbb{E} \left[ \int_0^1 \mathbbm{1}\left\{\pi(X,g(f(X)))>t\right\} dt \right] \\
        &= \int_0^1 \mathbb{P}\left[ \pi(X,g(f(X)))>t \right]dt \\
        &= \int_0^1 \mathbb{E} \left[\mathbb{P}\left[ \pi(X,g(f(X)))>t | X \right] \right]dt \\
        &= \int_0^1 \mathbb{E} \left[ \prod_{i=1}^{M} \mathbb{P}\left[ \pi(X,c_{i,1},c_{i,2})>t | X \right] \right]dt \label{thm1pf}
    \end{align} 
    where (\ref{thm1pf}) is obtained by the following observation:
    \begin{align}
        \mathbbm{1}\{ \pi(x,g(f(x))) > t \} &= \mathbbm{1}\left\{ \min_i \pi(x,c_{i,1},c_{i,2}) > t \right\} \\ \label{thm1pf2}
        &= \prod_{i=1}^M \mathbbm{1}\{ \pi(x,c_{i,1},c_{i,2}) > t \}.
    \end{align}
    Taking the average with respect to codewords drawn i.i.d. from $P_{\hat{X}\hat{Y}}$ and independent from $X$ and then taking infimum to find tightest bound gives us the right hand side of (\ref{eq:ach}).

%%%%%%%%%%%%%%%%%%%%%%%%%%%%%%%%%%%%%%%%%%%%%%%%%%%%%%%%%%%%%%%%%%%%%
%%%%%%%%%%%%%%%%%\
%\section{Proof of Corrolary \ref{cor1}} \label{apx:cor1}
%Setting $Y=X$ in (\ref{EqPi}) gives us:
%\begin{align}
%    \pi(x,\hat{x}_1,\hat{x}_2)  = \mathbbm{1}\left[\{d_1(x,\hat{x})> D_1 \} \cup \{ d_2(x,\hat{x}_2) > D_2\}\right]
%\end{align}
%Then for $0\leq t < 1$
%\begin{align}
%    \mathbb{P}\left[ \left. \pi(X,\hat{X}_1,\hat{X}_2)>t\right|X \right] =
%    \mathbb{E}\left[ \left. \pi(X,\hat{X}_1,\hat{X}_2) \right| X\right].
%\end{align}
%This allows us to simplify the integral in (\ref{eq:ach}) and obtain the desired expression.
%%%%%%%%%%%%%%%%%%%%%%%%%%%%%%%%%%%%%%%%%%%%%%%%%%%%
%%%%%%%%%%%%%%%%%%%%%%%%%%%%%%%%%%%%%%%%%%%%%%%%%%%%%%%%%%%%%%%%%%%%%
\section{Proof of Theorem \ref{conv}}  \label{apx:conv}
For this proof we drop the subscript of $\jmath_{X;\hat{X}^*,\hat{Y}^*}(x,y,\hat{x},\hat{y},D_1,D_2)$ for clarity. For an arbitrary single-shot joint lossy source code $(f,g)$, let $f(X)=U$ and $g(U) = (\hat{X},\hat{Y})$. Note that $U$ can only take values in $\{1,\ldots,M\}$. Define the event:
\begin{align}
     E : \{d_1(x,\hat{x})>D_1\} \cup \{d_2(y,\hat{y})>D_2\}.
\end{align}
Let the complement of $E$ be $\Bar{E}$. Then for any $\gamma \geq 0$:
\begin{align}
    \mathbb{P}[\jmath(X,Y,\hat{X}&,\hat{Y},D_1,D_2)\geq \log M + \gamma] = \\ &\mathbb{P}[\jmath(X,Y,\hat{X},\hat{Y},D_1,D_2) \geq \log M + \gamma , E] \\
    +&\mathbb{P}[\jmath(X,Y,\hat{X},\hat{Y},D_1,D_2) \geq \log M + \gamma , \Bar{E}] \\
    \leq& \mathbb{P}(E) \\
    +&\mathbb{P}[\jmath(X,Y,\hat{X},\hat{Y},D_1,D_2) \geq \log M + \gamma , \Bar{E}] \label{convTerm1}
\end{align}
Now we bound the second term in (\ref{convTerm1}). We can write it is as follows:
\begin{align}
\begin{split}
    \sum_{y\in\mathcal{Y}}\sum_{x\in\mathcal{X}}P_{XY}(x,y)&\sum_{u=1}^{M}P_{U|X}(u|x)\sum_{(\hat{x},\hat{y})} P_{\hat{X}\hat{Y}|U}(\hat{x},\hat{y}|u) \\
    &\cdot \mathbbm{1}\{\jmath(x,y,\hat{x},\hat{y},D_1,D_2) \geq \log M + \gamma\}\\
    &\cdot \mathbbm{1}\{d_1(x,\hat{x})\leq D_1 \}\mathbbm{1}\{d_2(y,\hat{y})\leq D_2 \}
\end{split}
\end{align}
\vspace{-0.5cm}
\begin{align}\label{convTerm2}
\begin{split}
    \leq \sum_{y\in\mathcal{Y}}\sum_{x\in\mathcal{X}}P_{XY}(x,y)\sum_{u=1}^{M}&P_{U|X}(u|x)\sum_{(\hat{x},\hat{y})} P_{\hat{X}\hat{Y}|U}(\hat{x},\hat{y}|u) \\
    &\cdot \frac{2^{\jmath(x,y,\hat{x},\hat{y},D_1,D_2) - \gamma}}{M} \\
    &\cdot 2^{\lambda_1(D_1-d_1(x,\hat{x}))} 2^{\lambda_2(D_2-d_2(y,\hat{y}))}
\end{split}
\end{align}
\vspace{-0.3cm}
\begin{align}
\begin{split}
    = \sum_{y\in\mathcal{Y}}\sum_{x\in\mathcal{X}}P_{XY}(x,y)\sum_{u=1}^{M}P_{U|X}(u|x)&\sum_{(\hat{x},\hat{y})} P_{\hat{X}\hat{Y}|U} (\hat{x},\hat{y}|u)\\
    &\cdot \frac{2^{\imath_{X;\hat{X}^*,\hat{Y}^*}(x;\hat{x},\hat{y})-\gamma}}{M}
\end{split}
\end{align}
\vspace{-0.3cm}
\begin{align}
\begin{split}
     = \frac{2^{-\gamma}}{M}\sum_{y\in\mathcal{Y}}\sum_{x\in\mathcal{X}}&P_{Y|X}(y|x)\sum_{u=1}^{M}P_{U|X}(u|x) \\ 
     &\cdot \sum_{(\hat{x},\hat{y})} P_{\hat{X}\hat{Y}|U}(\hat{x},\hat{y}|u) P_{X|\hat{X}^*\hat{Y}^*}(x|\hat{x},\hat{y}) 
\end{split}
\end{align}
\vspace{-0.4cm}
\begin{align}
\begin{split}
    \leq \frac{2^{-\gamma}}{M}\sum_{u=1}^{M}\sum_{(\hat{x},\hat{y})}P_{\hat{X}\hat{Y}|U}(\hat{x},\hat{y}|u)\sum_{x\in\mathcal{X}}&P_{X|\hat{X}^*\hat{Y}^*}(x|\hat{x},\hat{y})  \\
    &\cdot \sum_{y\in\mathcal{Y}}P_{Y|X}(y|x) \label{convTerm3}\\ 
\end{split} 
\end{align}
\begin{flalign}
\;\;\;\;= 2^{-\gamma} &&
\end{flalign}
where
\begin{itemize}
    \item (\ref{convTerm2}) is due to $\mathbbm{1}\{a \leq b\} \leq 2^{c(b-a)}$ for all $c\geq0$ and $a,b \in \mathbb{R}$.
    \item (\ref{convTerm3}) is obtained by bounding $P_{U|X}(u|x)\leq1$ for all $(x,u)$.
\end{itemize}
We get the desired result by picking the $\gamma$ that gives us best bound and taking the worst distribution to obtain a code independent bound .

%%%%%%%%%%%%%%%%%%%%%%%%%%%%%%%%%%%%%%%%%%%%%%%%%%%%%%%%%%%%%%%%%%%%%
%%%%%%%%%%%%%%%%%%%%%%%%%%%%%%%%%%%%%%%%%%%%%%%%%%%%%%%%%%%%%%%%%%%%%
%%%%%%%%%%%%%%%%%%%%%%%%%%%%%%%%%%%%%%%%%%%%%%%%%%%%%%%%%%%%%%%%%%%%%
\section{Proof of Theorem \ref{thmLogLoss}} \label{apx:loglossUniv}
 Suppose $D_2\geq D_{2,\min}$ and let the rate distortion achieving distribution for $R_2(D_2)$ to be $P_{\hat{Y}^\star|X}$. First consider the case $R_1(D_1) \geq R_2(D_2)$. Then by our assumption $R_2(D_{2,\mathrm{min}})\geq H(X)-D_1$ and continuity of rate-distortion function, there exists a $D_2' \leq D_2$ such that $R_2(D_2') = R_1(D_1)$. Denote the random variable achieving $R_2(D_2')$ by $Y'$ and let $X'(x) = \mathbb{P}[X=x|Y']$. Then with this choice of $(X',Y')$ we have:
\begin{align}
    \mathbb{E}[\Tilde{d}_2(X,Y')] &\leq D_2' \leq D_2 \\
    \mathbb{E}[d_1(X,X')] &= H(X|Y') = D_1. 
\end{align}
We also have
\begin{align}
    R(D_1,D_2) &\leq I(X;X',Y') \\
    &= I(X;Y') \\
    &= R_2(D_2') = R_1(D_1).
\end{align}

Next consider the alternative $R_2(D_2) \geq R_1(D_1)$. Then we have
\begin{align}
    H(X|\hat{Y}^\star) &\leq D_1
\end{align}
Then consider $Y'' = \hat{Y}^\star$ and $X''(x) = \mathbb{P}[X=x|Y'']$. The pair $(X'',Y'')$ satisfies the distortion constraints:
\begin{align}
    \mathbb{E}[\Tilde{d}_2(X,Y'')] &= \mathbb{E}[\Tilde{d}_2(X,\hat{Y})] \leq D_2 \\
    \mathbb{E}[d_1(X,X'')] &= H(X|Y'') \leq D_1. 
\end{align}
where $\Tilde{d}_2(x,\hat{y}) = \mathbb{E}[d_2(Y,\hat{y})|X=x]$. Then
\begin{align}
    R(D_1,D_2) &\leq I(X;X'',Y'') \\
    &= I(X;Y'') \\
    &= R_2(D_2).
\end{align}

Since $X'$ is a function of $Y'$ we have $Y-X-Y'-X'$ (Similarly for $X''$ and $Y''$). Note that we can always ensure that this function is one-to-one. Suppose for the moment that the function is many-to-one, i.e. $\exists y'_1,y'_2 \in \mathcal{Y}'$, $y_1'\neq y_2'$ such that $\mathbb{P}[X=x|Y'=y'_1] = \mathbb{P}[X=x|Y'=y'_2]$ for every $x\in\mathcal{X}$. Furthermore, without loss of generality assume $\Tilde{d}_2(x,y_1') \leq \Tilde{d}_2(x,y_2')$. Then we can combine $y_1'$ and $y_2'$ into a single symbol $y_1'$. This new random variable will have a lower expected distortion than $\mathbb{E}[\Tilde{d}_2(X,Y')]$ and have the same mutual information as $I(X;Y')$. Thus we can always have a $Y'$ that has a one-to-one relation with $X'$.

% Thus we can always ensure $Y-X-X'-Y'$ also holds.

%%%%%%%%%%%%%%%%%%%%%%%\
%%%%%%%%%%%%%%%%%%%%%%%%%%%%%%%%%%%%%%%%%%%%%%
\section{Proof of Corollary \ref{convLogLoss}} \label{apx:cor2}
Suppose $R_2(D_2)<H(X)-D_1$, then $R(D_1,D_2)=H(X)-D_1$ and in (\ref{eq:lamdas}), $\lambda_1=1$ and $\lambda_2=0$. Assume that $P_{\hat{X}^*\hat{Y}^*|X}$ achieves the minimum in (\ref{eq:RD}). By the one-to-one relationship:
\begin{align}
    \imath_{X;\hat{X}^*\hat{Y}^*}(x;\hat{x},\hat{y}) &= \log \frac{P_{X|\hat{X}^*\hat{Y}^*}(x|\hat{x}^*,\hat{y}^*)}{P_X(x)} \\
    &= \imath_X(x) - \log \frac{1}{P_{X|\hat{Y}^*}(x|\hat{y}^*)} \label{eq:cor21}
\end{align}
and
\begin{align}
    d_1(x,\hat{x}^*) = \log \frac{1}{\hat{x}^*(x)} = \log \frac{1}{P_{X|\hat{Y}^*}(x|\hat{y}^*)}. \label{eq:cor22}
\end{align}
combining (\ref{eq:cor21}) and (\ref{eq:cor22}) with Theorem \ref{conv} gives us (\ref{eq:cor2p1}). For the other case, $R_2(D_2)\geq H(X)-D_1$, we have $R(D_1,D_2)=R_2(D_2)$ hence $\lambda_1=0$. Together with the one-to-one relationship between $\hat{X}^*$ and $\hat{Y}^*$ we have $\jmath_{X;\hat{X}^*,\hat{Y}^*}(x,y,\hat{x},\hat{y},D_1,D_2) = \jmath_{X;\hat{Y}^\star}(x,y,\hat{y},D_2)$.

%%%%%%%%%%%%%%%%%%%%%%%%%%%%%%%%%%%%%%%%%%%%%%%%%%%%%%%%%%%%%%%%%%%%%
%%%%%%%%%%%%%%%%%%%%%%%%%%%%%%%%%%%%%%%%%%%%%%%%%%%%%%%%%%%%%%%%%%%%%
\section{Proof of Theorem \ref{Ach2}} \label{apx:alt}
    Proof is similar to that of Theorem \ref{ach} except for the encoding function where we replace the encoder in (\ref{enc1}) with (\ref{enc2}) below. Fix a reconstruction codebook $P_{\hat{Y}}$ independent of $X$ and $0\leq\epsilon'\leq1$. Define the encoding metric:
    \begin{align}
        \pi_{\epsilon'}(x,\hat{y}) = \mathbbm{1}\left\{\mathbb{P}[d(Y,\hat{y})>D_2|X=x]\leq\epsilon'\right\} \label{encmetric}
    \end{align}
    Pick a random codebook with $M$ codewords $\{c_1,\ldots,c_M\}$ generated independently from $P_{\hat{Y}}$. Given a source realization $x$, encoder simply sends the index of the codeword $f(x) = i^*$
    \begin{align}
        i^* \in \argmax_{i \in {1,\ldots,M}} \pi_{\epsilon'}(x,c_i) \label{enc2}
    \end{align}
    where ties are broken arbitrarily. After receiving the index, decoder outputs the codeword pair with given index $g(i)=c_i$. Then the excess distortion probability given $X$ is:
    \begin{align}
        \mathbb{P}[d_2(Y,g&(f(X)))>D_2|X] = \\
        &\mathbb{P}[d_2(Y,g(f(X)))>D_2|X,\pi_{\epsilon'}(X,g(f(X)))=1]\\
        &\cdot\mathbb{P}[\pi_{\epsilon'}(X,g(f(X)))=1|X] \\
        +&\mathbb{P}[d_2(Y,g(f(X)))>D_2|X,\pi_{\epsilon'}(X,g(f(X)))=0] \\
        &\cdot\mathbb{P}[\pi_{\epsilon'}(X,g(f(X)))=0|X] \\
        \leq& \epsilon'(1-\mathbb{P}[\pi_{\epsilon'}+(X,g(f(X)))=0|X])\\ 
        +& \mathbb{P}[\pi_{\epsilon'}(X,g(f(X)))=0|X] \label{thm4pf}
    \end{align}
    We make the same observation as in (\ref{thm1pf2}):
    \begin{align}
        \mathbbm{1}\{ \pi_{\epsilon'}(x,g(f(x))) = 0 \} &= \mathbbm{1}\left\{ \max_i \pi_{\epsilon'}(x,c_i) = 0 \right\} \\
        &= \prod_{i=1}^M \mathbbm{1}\{ \pi_{\epsilon'}(x,c_i) = 0 \} \\
        &=\prod_{i=1}^M \mathbbm{1}\{ \pi'(x,c_i) > \epsilon \}. \label{thm4pf2}
    \end{align}
     Since codewords are generated i.i.d. from $P_{\hat{Y}}$, let the random variables representing the codewords be $\hat{Y}_1,\ldots,\hat{Y}_M$. Taking the expectation given $X$ in (\ref{thm4pf2}) we get:
    \begin{align}
        \mathbb{E}\left[ \left. \prod_{i=1}^M \mathbbm{1}\{ \pi'(X,\hat{Y}_i) > \epsilon \} \right| X\right] &= \prod_{i=1}^M \mathbb{P}\left[ \left. \pi'(X,\hat{Y}_i) > \epsilon \right| X\right] \\
        &= \mathbb{P}\left[ \left. \pi'(X,\hat{Y}) > \epsilon \right| X\right]^M
    \end{align}
    Using this result in (\ref{thm4pf}) taking the infimum with respect to $P_{\hat{Y}}$ and $\epsilon'$ completes the proof.

%%%%%%%%%%%%%%%%%%%%%%%%%%%%%%%%%%%%%%%%%%%%%%%%%%%%%%%%%%%%%%%%%%%%%
%%%%%%%%%%%%%%%%%%%%%%%%%%%%%%%%%%%%%%%%%%%%%%%%%%%%%%%%%%%%%%%%%%%%%
%%%%%%%%%%%%%%%%%%%%%%%%%%%%%%%%%%%%%%%%%%%%%%%%%%%%%%%%%%%%%%%%%%%%%
%%%%%%%%%%%%%%%%%%%%%%%%%%%%%%%%%%%%%%%%%%%%%%%%%%%%%%%%%%%%%%%%%%%%%
\section{Proof of Theorem \ref{thm:achLog}} \label{apx:achLog}
    The proof uses a random coding method similar to that of Theorem \ref{Ach2}, together with a binning argument as developed in \cite{shkel2017}. Fix a $P_{\hat{Y}}$, $\gamma\geq0$ and $0\leq\epsilon\leq1$. Pick a random codebook with $M$ codewords $\{c_1,\ldots,c_M\}$ generated i.i.d. from $P_{\hat{Y}}$. Using the same encoding metric in (\ref{encmetric}) define for each $x\in\mathcal{X}$ the following set:
    \begin{align}
        \mathcal{C}_x = \left\{i:\argmax_{i \in {1,\ldots,M}} \pi_{\epsilon'}(x,c_i)\right\}
    \end{align}
    $f(x)=i\in\mathcal{C}_x$ chosen uniformly. Furthermore create a virtual bin uniformly $b(x)=j\in\{1,\ldots,L\}$ where $L=\lfloor{2^{D_1}}\rfloor$. The reasoning here is that given a distortion constraint $D_1$, we need to spare at least $2^{-D_1}$ probability to $\hat{x}(x)$ so that $d_1(x,\hat{x}) \leq D_1$. Hence a single $\hat{x}$ allows us to cover $\lfloor{2^{D_1}}\rfloor$ elements \cite{shkel2017}. Encoder sends the index $f(x)=i$. Decoder upon receiving the index $i$, for the indirect part sets the output as $g_2(i)=c_i$. For the direct part, define the following set:
    \begin{align}
        \mathcal{B}(i) = \left\{ x \in \bigcup_{j:|f^{-1}(i)\cap b^{-1}(j) \cap \mathcal{L}|=1} f^{-1}(i)\cap b^{-1}(j) \cap \mathcal{L} \right\}
    \end{align}
    where
    \begin{align}
        f^{(-1)}(i) &= \{x:f(x)=i\}, \\
        b^{(-1)}(j) &= \{x:b(x)=j\}, \\
        \mathcal{L} &= \{x:\imath_{X}(x) \leq D_1 + \log M - \gamma\}.
    \end{align}
    Now define the reconstruction $g_1(i)=P_i$ as follows. If $\mathcal{B}(i) \neq \emptyset$ then:
    \begin{align}
        P_i(x) = \begin{cases} \frac{1}{|\mathcal{B}(i)|}, \;\; &x\in\mathcal{B}(i) \\
        0, \;\; &\mathrm{otherwise}
        \end{cases}
    \end{align}
    and $P_i=P_X$ if $\mathcal{B}(i) = \emptyset$. Observe that if $x\in\mathcal{B}(i)$ then
    \begin{align}
        d_1(x,g_1(f(x))) = \log|\mathcal{B}(i)| \leq \log L \leq D_1
    \end{align}
    Hence it satisfies the distortion constraint. In light of this we begin analyzing the performance. Assume for the moment that $(x_0,y_0)$ is the realization of the source and $f(x_0)=i$, $b(x_0)=j$ the excess distortion event is contained in union of the following error events:
    \begin{align}
        E_1 &: d_2(y_0,c_i)>D_2  \\ 
        E_2 &: \imath_X(x_0)>D_1+\log M -\gamma \\
        \begin{split}
        E_3 &: \exists x\neq x_0, \\
        &\;\;\;f(x)=i \cap b(x)=j \cap \imath_X(x)\leq D_1+\log M -\gamma    
        \end{split}
    \end{align}
    $\mathbb{P}[E_1]$ can be bounded above by using Theorem \ref{Ach2}. $\mathbb{P}[E_2] = \mathbb{P}[\imath_X(X)>D_1+\log M -\gamma]$ independent of the codebook. To bound $\mathbb{P}[E_3]$, suppose for the moment $\{c_1,\ldots,c_M\}$ and $x_0$ are fixed. Then
    \begin{align}
        \begin{split}
            \mathbb{P}[E_3] &\leq \sum_{x\neq x_0}\mathbb{P}[f(x)=f(x_0) \cap b(x)=b(x_0)] \\
            &\;\;\;\cdot \mathbbm{1}\{\imath_X(x)\leq D_1+\log M -\gamma\}
        \end{split} \\
        &= \frac{1}{L}\sum_{x\neq x_0}\frac{|\mathcal{C}_x \cap \mathcal{C}_{x_0}|}{|\mathcal{C}_x|| \mathcal{C}_{x_0}|}\mathbbm{1}\{\imath_X(x)\leq D_1+\log M -\gamma\} \\
        &\leq \frac{1}{L|C_{x_0}|}\sum_{x\neq x_0} \{\imath_X(x)\leq D_1+\log M -\gamma\} \\
        &\leq \frac{1}{L|C_{x_0}|}M2^{D_1}2^{-\gamma} \\
        &\leq \frac{M}{|\mathcal{C}_{x_0}|}2^{1-\gamma}. \label{eq:cx0}
    \end{align}
    Next we average (\ref{eq:cx0}) with respect to the codebook given $X=x_0$. In order to do that first consider $|\mathcal{C}_{x_0}|$:
    \begin{align}
        |C_{x_0}|= \begin{cases}M, \;\;\;\;\;\;\;\;\;\;\;\;\;\;\; \sum_{i=1}^{M} \mathbbm{1}\left\{\pi'(x_0,c_i)\leq\epsilon'\right\}=0 \\
            \sum_{i=1}^{M} \mathbbm{1}\left\{\pi'(x_0,c_i)\leq\epsilon'\right\},\;\;\;\;\;\;\;\;\;\; \mathrm{otherwise}
        \end{cases}
    \end{align}
Now denote the random variables that generate the $c_i$ as $\hat{Y}_i$ for $1\leq i \leq M$. Since $\hat{Y}_i$ are i.i.d., given $X$, the random variables $\mathbbm{1}\{\pi'(x,\hat{Y}_i)\leq \epsilon'\}$ are i.i.d. Bernoulli random variables with parameter $1-\eta(\epsilon')$. Thus their sum will be a Binomial random variable. The distribution of $|C_X|$ given $X$ will be close to this Binomial except for the end points since $|\mathcal{C}_X|$ cannot be 0. Hence for $1\leq k < M$: 
    \begin{align}
        \mathbb{P}[|\mathcal{C}_{X}|=k|X] = {M \choose k}(1-\eta(\epsilon'))^k\eta(\epsilon')^{M-k}
    \end{align}
    and
    \begin{align}
        \mathbb{P}[|\mathcal{C}_{X}|=M|X] = \eta(\epsilon')^M + (1-\eta(\epsilon'))^M.
    \end{align}
    This allows us to calculate the desired expectation:
    \begin{align}
    \begin{split}
        \mathbb{E}\left[ \left. \frac{1}{|C_X|}\right|X \right] =& \frac{1}{M}(\eta(\epsilon')^M + (1-\eta(\epsilon'))^M) \\
        +& \sum_{k=1}^{M-1} \frac{1}{k} {M \choose k} (1-\eta(\epsilon'))^k\eta(\epsilon')^{M-k}
    \end{split}
    \end{align}
    Finally taking the expectation with respect to $X$ gives us the desired bound for $P[E_3]$. Using union bound on $\mathbb{P}[\cup_{i=1}^3E_i]$ and optimizing with respect to $\epsilon'$, $\gamma$ and $P_{\hat{Y}}$ completes the proof.
    
\bibliographystyle{IEEEtran}
\bibliography{references}

\end{document}